\documentclass[11pt,draftcls,onecolumn]{IEEEtran}
\usepackage{psfrag,cite,amsmath,graphicx,epsfig,latexsym,subfigure,color}
\usepackage[scriptsize,bf]{caption}
\usepackage{color}
\usepackage{amssymb}
\usepackage{setspace}

\DeclareMathOperator{\argmin}{\mbox{argmin}}
\def\bn{\hfill \\ \smallskip\noindent}
\def\argmin{\mathop{\rm argmin}}

\def\vx{x}

\def\dom{\mbox{dom\,}}

\def\prox{\mbox{prox}}

\newcommand{\beq}{\begin{equation}}
\newcommand{\eeq}{\end{equation}}
\newcommand{\st}{{\rm s.t.}}

\newcommand{\cC}{{\mbox{ $\mathcal{C}$}}}

\begin{document}
\def\pn {\par\smallskip\noindent}
\def \bn {\hfill \\ \smallskip\noindent}
\newcommand{\fs}{f_1,\ldots,f_s}
\newcommand{\f}{\vec{f}}
\newcommand{\hx}{\hat{x}}
\newcommand{\hy}{\hat{y}}
\newcommand{\barhx}{\bar{\hat{x}}}
\newcommand{\vecx}{x_1,\ldots,x_m}
\newcommand{\xoy}{x\rightarrow y}
\newcommand{\barx}{{\bar x}}
\newcommand{\bary}{{\bar y}}
\newtheorem{theorem}{Theorem}[section]
\newtheorem{lemma}{Lemma}[section]
\newtheorem{corollary}{Corollary}[section]
\newtheorem{proposition}{Proposition}[section]
\newtheorem{definition}{Definition}[section]
\newtheorem{claim}{Claim}[section]
\newtheorem{remark}{Remark}[section]

\def\br{\break}
\def\smskip{\par\vskip 5 pt}
\def\proof{\bn {\bf Proof.} }
\def\QED{\hfill{\bf Q.E.D.}\smskip}
\def\qed{\quad{\bf q.e.d.}\smskip}

\newcommand{\cM}{\mathcal{M}}
\newcommand{\cJ}{\mathcal{J}}
\newcommand{\cT}{\mathcal{T}}
\newcommand{\bx}{\mathbf{x}}
\newcommand{\bp}{\mathbf{p}}
\newcommand{\bz}{\mathbf{z}}

\title{A Distributed, Asynchronous and Incremental Algorithm for Nonconvex Optimization: An ADMM Based Approach}
\author{{Mingyi Hong}\thanks{M. Hong is with the Department of Industrial and Manufacturing Systems Engineering (IMSE), Iowa State University, Ames, IA 50011, USA. Email: \texttt{mingyi@iastate.edu}}}
\maketitle
\begin{abstract}
The alternating direction method of multipliers (ADMM) has been popular for solving many signal processing problems, convex or nonconvex. In this paper, we study an asynchronous implementation of the ADMM for solving a nonconvex nonsmooth optimization problem, whose objective is the sum of a number of component functions. The proposed algorithm allows the problem to be solved in a distributed, asynchronous and incremental manner. First, the component functions can be distributed to different computing nodes, who perform the updates asychronously without coordinating with each other. Two sources of asynchrony are covered by our algorithm: one is caused by the heterogeneity of the computational nodes, and the other arises from  unreliable communication links.
Second, the algorithm can be viewed as implementing an incremental algorithm where at each step the (possibly delayed) gradients of only a subset of component functions are updated. We show that when certain bounds are put on the level of asynchrony, the proposed algorithm converges to the set of stationary solutions (resp. optimal solutions) for the nonconvex (resp. convex) problem. To the best of our knowledge, the proposed ADMM implementation can tolerate the highest degree of asynchrony, among all known asynchronous variants of the ADMM. Moreover, it is the first ADMM implementation that can deal with nonconvexity and asynchrony at the same time.
\end{abstract}

\section{Introduction}
Consider the following nonconvex and nonsmooth problem
\begin{align}\label{eq:consensus}
\begin{split}
\min&\quad f(x):=\sum_{k=1}^{K}g_k(x)+h(x)\\
\st&\quad x\in X
\end{split}
\end{align}
where $g_k$'s are a set of smooth, possibly nonconvex functions; $h(x)$ is a convex nonsmooth regularization term. In this paper we consider the scenario where the component functions $g_k$'s are located at different distributed computing nodes. We seek an algorithm that is capable of computing high quality solutions for problem \eqref{eq:consensus} in a distributed, asynchronous and incremental manner.


Dealing with asynchrony is a central theme in designing distributed algorithms. Indeed, often in a completely decentralized setting, there is no clock synchronization, little coordination among the distributed nodes, and minimum mechanism to ensure reliable communication. Therefore an ideal distributed algorithm should be robust enough to handle different sources of asynchrony, while still producing high quality solutions in a reasonable amount of time. Since the seminal work of Bertsekas and Tsitsiklis \cite{Bertsekas_Book_Distr, tsitsiklis86}, there has been a large body of literature focusing on asynchronous implementation of various distributed schemes; see, e.g., \cite{zhu10, nedic01, scutari08c, scutari07asyn, liao15icassp} for the developments by the optimization and signal processing communities. In \cite{nedic01}, an incremental and asynchronous gradient-based algorithm is proposed to solve a convex problem, where at each step certain outdated gradients can be used for update. In \cite{scutari07asyn, scutari08c}, the authors show that the well-known iterative water-filling algorithm \cite{yu02a,scutari08a} can be implemented in a totally asynchronous manner, as long as the interference among the users are weak enough.

The recent interest in optimization and machine learning for problems with massive amounts of data introduces yet another compelling reason for dealing with asynchrony; see \cite[Chapter 10]{mass_data13}. When large amounts of data are distributedly located at computing nodes, local computations can be costly and time consuming. If synchronous algorithms are used, then the slowest nodes can drag the performance of the entire system. To make distributed learning algorithms scalable and efficient, the machine learning community has also started to deal with asynchrony; see recent results in \cite{Ho13, Li13Distributed, niu11, liu14asynchronous, Agarwal11, langford09}. For example in \cite{liu14asynchronous}, an asynchronous randomized block coordinate descent method is developed for solving convex block structured problem, where the per-block update can utilize delayed gradient information. In \cite{Agarwal11}, the authors show that it is also possible to tolerate asynchrony in stochastic optimization. Further, they prove that the rate of the convergence is more or less independent of the maximum allowable delay, which is an improvement over earlier results in \cite{nedic01}.

In this paper, we show that through the lens of the ADMM method, the nonconvex and nonsmooth problem \eqref{eq:consensus} can be optimized in an asynchronous, distributed, and incremental manner. The ADMM, originally developed in early 1970s \cite{ADMMGlowinskiMorroco,ADMMGabbayMercier}, has been extensively studied in the last two decades \cite{Eckstein89, EcksteinBertsekas1992, Glow84, bertsekas97, BoydADMM, ADMMlinearYin, HeYuan2012, Monteiro13, HongLuo2012ADMM}. It is known to be effective in solving large-scale linearly constrained convex optimization problems. Its application includes machine learning, computer vision, signal and image processing, networking, etc; see  \cite{Yin:2008:BIA:1658318.1658320, Yang09TV, Scheinberg10inverse, Schizas08, Feng14, liao14sdn}. However, despite various successful numerical attempts (see, e.g., \cite{zhang10ADMM_NMF, sun14,wen12,zhang14,Forero11,Ames13LDA,jiang13ADMM,Liavas14,Shen:2014}), little is known about whether ADMM is capable of handling nonconvex optimization problems, or whether it can be used in an asynchronous setting.
There are a few recent results that start to fill these gaps. Reference \cite{hong14nonconvex_admm} shows that the ADMM converges when applied to certain nonconvex consensus and sharing problems, provided that the stepsize is chosen large enough. However it is not clear whether asynchrony will destroy the convergence. Reference \cite{zhang14} proposes an asynchronous implementation for convex global consensus problem, where the distributed worker nodes can use outdated information for updates. Two conditions are imposed on the protocol, namely the {\it partial barrier} and {\it bounded delay}. The algorithm cannot deal with the  asynchrony cause by loss/delay in the communication link, nor does it cover nonconvex problems. In \cite{Wei13, chang14} randomized versions of ADMM are proposed for consensus problems, where the nodes are allowed to be randomly activated for updates. We note that the algorithms in \cite{Wei13, chang14}  still require the nodes to use up-to-date information whenever they update, therefore they are more in line with randomized algorithms than asynchronous algorithms. Further, it is not known whether the analysis carries over to the case when the problem is nonconvex.

The algorithm proposed in this work is a generalization of the flexible proximal ADMM algorithm proposed in \cite[Section 2.3]{hong14nonconvex_admm}. The key feature of the proposed algorithm is that it can deal with asynchrony arises from the heterogeneity of the computing nodes as well as the loss/delay caused by unreliable communication links. The basic requirement here is that  the combined effects of these sources leads to a  bounded delay on the component gradient evaluation, and that the stepsize of the algorithm is chosen appropriately. Further, we show that the framework studied here can be viewed as an (possibly asynchronous) incremental scheme for nonconvex problem, where at each iteration only a subset of (possibly delayed) component gradients are updated. To the best of our knowledge, asynchronous incremental schemes of this kind hasn't been studied in the literature; see \cite{Sra12,Razaviyayn13Stochastic,Mairal14incremental} for recent works on {\it synchronous} incremental algorithm for nonconvex problems.

\section{The ADMM-based Framework}\label{sec:algorithm}
\subsection{Preliminary}
Consider the optimization problem \eqref{eq:consensus}. In many practical applications, $g_k$'s need to be handled by a single distributed node, such as a thread or a processor, which motivates the so-called global consensus formulation \cite[Section 7]{BoydADMMsurvey2011}. Suppose there is a master node and $K$ distributed nodes available. Let us introduce a set of new variables $\{x_k\}_{k=1}^{K}$, and transform problem \eqref{eq:consensus} to the following linearly constrained problem
\begin{align}\label{eq:consensus:admm}
\begin{split}
\min&\quad \sum_{k=1}^{K}g_k(x_k)+h(x)\\
\st&\quad x_k=x,\;\forall~k=1,\cdots,K, \quad x\in X.
\end{split}
\end{align}
The augmented Lagrangian function is given by
\begin{align}\label{eq:lagrangian:consensus}
\begin{split}
L(\{x_k\}, x; y)&=\sum_{k=1}^{K}g_k(x_k)+h(x)+\sum_{k=1}^{K}\langle y_k, x_k-x\rangle\\
&\quad+\sum_{k=1}^{K}\frac{\rho_k}{2}\|x_k-x\|^2,
\end{split}
\end{align}
where $\rho_k>0$ is some constant, and $y:=\{y_1,\cdots, y_K\}$.  Applying the vanilla ADMM algorithm, listed below in \eqref{eq:admm:vanilla},  one obtains a distributed solution where each function $g_k$ is only handled by a single node $k$ at any iteration $t=0,1, 2, \ldots$:
\begin{align}\label{eq:admm:vanilla}
\begin{split}
x^{t+1}&=\arg\min_{x\in X}\; L(\{x^t_k\}, x; y^t) \\
&=\arg\min_{x\in X}\; h(x)+\sum_{k=1}^{K}\langle y^t_k, x^t_k-x\rangle +\sum_{k=1}^{K}\frac{\rho_k}{2}\|x^t_k-x\|^2\\
x_k^{t+1}&=\arg\min_{x_k}\; L(\{x_k\}, x^{t+1}; y^t),\\
 &= \arg\min _{x_k} \; g_k(x_k) + \langle y^t_k, x_k-x^{t+1}\rangle+\frac{\rho_k}{2}\|x_k-x^{t+1}\|^2,\; \forall~k \\
y^{t+1}_k&= y^t_k + \rho_k (x^{t+1}_k-x^{t+1}), \; \forall ~k.
\end{split}
\end{align}
Under suitable conditions the algorithm converges to the set of stationary solutions of \eqref{eq:consensus}; see \cite{hong14nonconvex_admm}.

At this point, it is important to note that the algorithm described in \eqref{eq:admm:vanilla} uses a synchronous protocol, that is
 \begin{itemize}
 \item [a)] The set of agents that are selected to update at each iteration act in a coordinated way;
 \item [b)] There is no communication delay and/or loss between the agents and the master node;
 \item [c)] All local updates are performed assuming that the most up-to-date information is available.
 \end{itemize}
However, in many practical large-scale networks, these assumptions are hardly true. Nodes may have different computational capacity, or they may be assigned jobs that have different computational requirements. Therefore the time consumed to complete local computation can vary significantly among the nodes. This makes them difficult to coordinate with each other in terms of when to update, which information to use for the update and so on. Further, the communication links between the distributed and the master nodes can have delays or may even be lossy.

Additionally, we want to mention that in certain machine learning and signal processing problems when there is a large number of component functions, it is desirable that the algorithm is {\it incremental}, meaning at each iteration only a {\it subset} of $g_k's$ are used for update; see \cite{Schmidt13,Defazio14,Sra12,Razaviyayn13Stochastic,Mairal14incremental}. Clearly the vanilla ADMM described in \eqref{eq:admm:vanilla} does not belong to this type of algorithm.

\subsection{The Proposed Algorithm}
There are two key features that we want to build into the ADMM-based algorithm. One is to allow the nodes to use staled information for local computation, as long as such information is not ``too old" (this notion will be made precise shortly). This enables the nodes to have varying update frequency, therefore faster nodes do not need to wait for the slower ones. The other feature is to take into account scenarios where the communication links among the node are lossy or have delays. Below we give a high level description of the proposed scheme.

Suppose there is a master node and $K$ distributed nodes in the system. Let the index $t=1,\cdots$ denote the total number of updates that have been performed on the variable $x$. The master node takes care of updating all the primal and dual variables, while the distributed nodes compute the gradients for each component function $g_k$. At each iteration $t+1$, the master node first updates $x$. Then it waits a fixed period of time, collects a few (possibly staled) gradients of component functions returned by a subset of local nodes $\cC^{t+1}\subseteq\{1,\cdots, K\}$, then proceed to the next round of update.  On the other hand, each node $k$ is in charge of a local component function $g_k$. Based on the copy of $x$ passed along by the master node, node $k$ computes and returns the gradient of $g_k$ to the master node. Note that for data intensive applications, the computation of the gradient can be time consuming. Also there can be delays of communication between two different nodes in the network. Therefore there is no guarantee that during the period of computation and communication of the gradient of $g_k$, the $x$ variable at the master node will always remain the same.

To characterize the possible delay involved in the computation and communication,  we define a new sequence $\{t(k)\}$, where each $t(k)$ represents the index of the copy of $x$ that evaluates the $\nabla g_k$ used by the master node at iteration $t$.

The proposed algorithm, named Asynchronous Proximal ADMM (Async-PADMM), is given in the following table.
\begin{center}
\fbox{
\begin{minipage}{3.4 in}
\smallskip
\centerline{\bf Algorithm 1. The Async-PADMM for Problem \eqref{eq:consensus:admm}}
\smallskip\small
S1) At each iteration $t+1$, compute:
\begin{align}\label{eq:x_update_prox}
\begin{split}
x^{t+1}&={\rm arg}\!\min_{x\in X}\; L(\{x_{k}^{t}\}, x;  y^{t})\\
&=\prox_{\iota(X)+h}\left[\frac{\sum_{k=1}^{K}\rho_k x_{k}^{t}+\sum_{k=1}^{K}y_{k}^{t}}{\sum_{k=1}^{K}\rho_k}\right].
\end{split}
\end{align}
S2) Pick a set $\cC^{t+1}\subseteq\{1,\cdots, K\}$, for all $k\in\cC^{t+1}$, update index $[t+1](k)$; for all $k\notin\cC^{t+1}$, let $[t+1](k) = [t](k)$
S3) Update $x_k$ by solving:
\begin{align}\label{eq:x_k_update_prox}
x^{t+1}_k&=\arg\!\min_{x_k} \; \langle\nabla g_k(x^{[t+1](k)}), x_k-x^{t+1}\rangle+\langle y^{t}_k, x_k-x^{t+1}\rangle\nonumber\\
&+\frac{\rho_k}{2}\|x_k-x^{t+1}\|^2, \; \forall~k=1,\cdots, K.
\end{align}
S4) Update the dual variable:
\begin{align}\label{eq:y_update_prox}
y^{t+1}_k=y_k^{t}+\rho_k\left(x^{t+1}_k-x^{t+1}\right), \; \forall~k=1,\cdots, K.
\end{align}
\end{minipage}
}
\end{center}
In Algorithm 1, we have used the proximity operator, which is defined below. Let $h:\dom(h)\mapsto \mathbb{R}$ be a (possibly nonsmooth) convex function. For every $\vx\in\dom(h)$, the \emph{proximity operator}
of $h$ is defined as \cite[Section 31]{Rockafellar70}
\begin{align}\label{eq:prox}
\prox_{h}(\vx)={\argmin_{u}}\;\;h(u)+\frac12\|\vx-u\|^2.
\end{align}

We note that in Step S2, $\cC^{t+1}$ defines the subset of component functions whose gradients have arrived during iteration $t+1$; again $[t+1](k)$ is the index of the copy of $x$ that evaluates the $\nabla g_k$ used by the master node at iteration $t+1$. For those component functions without new gradient information available, the old gradients will continue to be used (indeed, note that we have for all $k\notin\cC^{t+1}$, $[t+1](k) = [t](k)$). In Step S3, all the variables $\{x_k\}$, regardless $k\in\cC^{t+1}$ or not, are updated according to the following gradient-type scheme:
\begin{align}
x^{t+1}_k=x^{t+1}-\frac{1}{\rho_k}\left(\nabla g_k(x^{[t+1](k)})+y^t_k\right).\label{eq:gradient}
\end{align}
Despite the fact that the gradients of all the component functions are used at each step $t+1$, only a subset of them (i.e., thosed indexed by $\cC^{t+1}$) differ from those at the previous iteration. Therefore the algorithm can be classified as {\it incremental} algorithm; see \cite{Schmidt13,Defazio14} for related incremental algorithms for convex problems.

To highlight the asynchronous aspect of the algorithm, below we present an equivalent version of Algorithm 1, from the perspective of the distributed nodes and the master node, respectively. We use $r_k$, $k=1,\cdots, K$ to denote the clock at node $k$, and use $r_0$ to denote the clock at the master node.

\begin{center}
\fbox{
\begin{minipage}{3.4 in}
\smallskip
\centerline{\bf Algorithm 1(a). Async-PADMM at the Master Node}
\smallskip\small
S0) {\bf Set} $r_0=1$, initialize $\{x^1_k, y^1_k\}$, $x^1$.\\
S1) {\bf Update $x$}:
\begin{align}\label{eq:x_update_prox}
x^{r_0+1}&={\rm arg}\!\min_{x\in X}\; L(\{x_{k}^{r_0}\}, x;  y^{r_0})\nonumber\\
&=\prox_{\iota(X)+h}\left[\frac{\sum_{k=1}^{K}\rho_k x_{k}^{r_0}+\sum_{k=1}^{K}y_{k}^{r_0}}{\sum_{k=1}^{K}\rho_k}\right].
\end{align}
S2) {\bf Broadcast} $x^{r_0+1}$ to all agents. \\
S3) {\bf Wait} for a fixed period of time. \\
S4) {\bf Collect} a set $\cC^{r_0+1}\subseteq\{1,\cdots, K\}$ of new local gradients, denoted as $\{z^{r_0+1}_k\}_{k\in\cC^{r_0+1}}$, arrived during S3). If multiple gradients arrive from the same node, pick the one with the smallest local time stamp.\\
S5) {\bf Let} $\nabla g^{r_0+1}_k = z^{r_0+1}_k$, $\forall~k\in\cC^{r_0+1}$.\\
S6) {\bf Let} $\nabla g^{r_0+1}_k = \nabla g^{r_0}_k$, $\forall~k\notin\cC^{r_0+1}$.\\
S7) {\bf Compute}
\begin{align}
\begin{split}
x_k^{r_0+1}&=x^{r_0+1}-\frac{1}{\rho_k}\left(\nabla g_k^{r_0+1}+y^{r_0}_k\right), \; \forall~k\\
y^{r_0+1}_k&=y_k^{r_0}+\rho_k\left(x^{r_0+1}_k-x^{r_0+1}\right), \; \forall~k.
\end{split}
\end{align}
S8) {\bf Set} $r_0=r_0+1$, go to step S1).
\end{minipage}}
\end{center}

\begin{center}
\fbox{
\begin{minipage}{3.4 in}
\smallskip
\centerline{\bf Algorithm 1(b). The Async- PADMM at Node $k$}
\smallskip\small
S0) {\bf Set} $r_k=1$.\\
S1) {\bf Wait} until a new $x$ is arrived, mark it as $x^{r_k}$.\\
S2) {\bf Compute the gradient} $\nabla g_k(x^{r_k})$.\\
S3) {\bf Send} $\nabla g_k(x^{r_k})$ and the local time stamp $r_k$ to the master node.\\
S4) {\bf Set} $r_k=r_k+1$, go to step S1).
\end{minipage}
}
\end{center}

It is not hard to see that the scheme described here is equivalent to Algorithm 1, except that in Algorithm 1 every step is measured using the clock at the master node. We have the following remarks regarding to the above algorithm descriptions.

\begin{remark}
{\it (Blocking Events)} There is a minimal number of blocking events for both the master node and the distributed agents. In Algorithm 1(a), the master node only needs to wait for a given period of time in step S3). After the waiting period, it collects the set of new gradients that has arrived during that period. Note that $\cC^{r_0+1}$ is allowed to be an empty set, meaning the master node is {\it not} blocking on the arrival of any local gradients. Similarly, each node $k$ does not need to wait for the rest of the agents to perform computation: once it obtains a new copy of $x^{r_k+1}$ the computation starts immediately. As soon as the computation is done node $k$ can send out the new gradient, without checking whether that gradient has arrived at the master node. Admittedly, in Step S1 of Algorithm 1(b), node $k$ needs to wait for a new $x$, but this is reasonable because otherwise there is nothing it can do.
\end{remark}

\begin{remark}
{\it (Characterization on the Delays)}
{The proposed algorithm allows communication delays and packet loss between the master and the distributed nodes. For example, the vector $x^{t+1}$ broadcasted by the master node may arrive at the different distributed nodes at different time instances; it may even arrive at a given node out of order, i.e., $x^{t+1}$ arrives before $x^{t}$. Further, $x^{t+1}$ may get lost during the transmission and never reaches a given node. All these scenarios can happen in the reverse communication direction as well. Comparing Algorithm 1 and Algorithm 1(a)--(b), we see that if $k\in\cC^{t+1}$, then the difference $(t+1)-[t+1](k)$ is the total computation time and the round-trip communication delay, starting from broadcasting $x^{[t+1](k)}$ until the updated $\nabla g_k(x^{[t+1](k)})$ is received by the master node. If $k\notin\cC^{t+1}$, then the difference $(t+1)-[t+1](k)$ is the number of times that the gradient $\nabla g_k (x^{[t+1](k)})$ has been used so far (or equivalently the number of iterations since the last gradient from node $k$ has arrived). Clearly, when there is no delay at all , then the system is synchronous and we have $[t+1](k)=t+1$.} In Fig. \ref{fig:time}, we illustrate the relationship $t$ and $t(k)$, and different types of asynchronous events covered by the algorithm.
\begin{figure*}
\center
\includegraphics[width=0.6\linewidth]{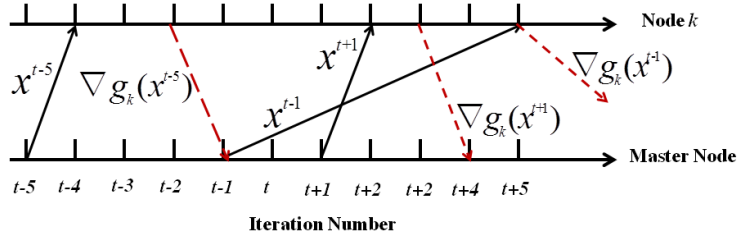}
\caption{Illustration of the asynchronous update rules. In this figure, $[t-1](k)=t-5$, while $[t+4](k)=t+1$. Note that $x^{t-1}$ arrives at node $k$ after $x^{t+1}$ does. Also note that the copies of $x$ broadcasted by the master node from $t-4$ to $t-2$ are not received by node $k$, either because node $k$ is busy computing the gradient, or because the copies get lost during communication. Finally, it may happen that some gradient, say $\nabla g_k (x^{t-1})$, never arrives at the master node.}\label{fig:time}
\end{figure*}
\end{remark}

\begin{remark}
{\it (Connection to Existing Algorithms)}
To the best of our knowledge, the proposed algorithm can tolerate the highest degree of asynchrony, among all known asynchronous variants of ADMM. For example, the scheme proposed in \cite{zhang14} corresponds to the case where there is no communication delay or loss (all messages sent are received instantaneously by the intended receiver). It is not clear whether the scheme in \cite{zhang14} can be generalized to our case\footnote{In fact, no proof is provided in \cite{zhang14}. Therefore it becomes difficult to see whether it is possible to extend their analysis.}. The schemes proposed in \cite{Wei13} and  \cite{chang14}  require the nodes to use the most up-to-date information, hence hardly asynchronous. The second major difference with the existing literature is about the tasks performed by the distributed nodes: in \cite{chang14,Wei13,zhang14} each node directly optimizes the augmented Lagrangian, while here each node  computes the gradient of their respective component functions. The third difference is on the assumptions made on problem \eqref{eq:consensus}: the schemes in \cite{chang14,Wei13,zhang14} handle convex problem but {\it each} component function $g_i$ can be nonsmooth, while we can handle nonconvex functions, but there can be only a single nonsmooth function $h$ (see Assumption A1 below). The fourth difference is on the assumed network topology:  the schemes in \cite{chang14,Wei13} deal with general topology, where nodes are interconnected according to certain graphs; our work and \cite{zhang14} are restricted to the ``star" network topology where all distributed nodes communicate directly with the master node.
\end{remark}

\begin{remark}
{\it (Incrementalism)} Algorithm 1 can be viewed as an {\it incremental} algorithm, as long as each $|\cC^{t+1}|$ is a strict subset of $\{1,\cdots K\}$, in which case the gradients of only a {\it subset} of component functions are updated.  This is in the same spirit of several recent incremental algorithms for convex problems \cite{Schmidt13,Defazio14}, despite the fact that our algorithm has a different form, and we can further handle nonconvexity and asychrony.

It is worth noting that Algorithm 1 can be modified to resemble the more traditional incremental algorithm {\cite{bertsekas2000incremental}}, where each iteration only those variables with ``fresh" gradients are updated. That is,  steps S3 and S4 are replaced with the following steps:
\begin{center}
\fbox{
\begin{minipage}{3.4 in}
\smallskip\small
S3)' Update $x_k$ by solving:
\begin{align}
x^{t+1}_k&=\arg\!\min_{x_k} \; \langle\nabla g_k(x^{[t+1](k)}), x_k-x^{t+1}\rangle+\langle y^{t}_k, x_k-x^{t+1}\rangle\nonumber\\
&+\frac{\rho_k}{2}\|x_k-x^{t+1}\|^2, \; \forall~k\in\cC^{t+1}.\nonumber
\end{align}
S4)' Update the dual variable:
\begin{align}
y^{t+1}_k=y_k^{t}+\rho_k\left(x^{t+1}_k-x^{t+1}\right), \; \forall~k\in\cC^{t+1}.\nonumber
\end{align}
\end{minipage}
}
\end{center}

However, we found that this variant leads to much more complicated analysis \footnote{To analyze this version, we need to define a few additional sequences, one for each node $k$, to characterize the iteration indices in which each component variable $x_k$ is updated. We will also need to impose that the $x_k$'s are updated often enough; see \cite[Chapter 7]{Bertsekas_Book_Distr}.}, stringent requirement on the range of stepsizes $\rho_k$'s, and most importantly, slow convergence. Therefore we choose not to discuss the related variants in the paper. We also note that recent works in incremental-type algorithms for solving \eqref{eq:consensus} either do not deal with nonconvex problem \cite{Schmidt13,Defazio14}, or they do not consider asynchrony \cite{Sra12,Razaviyayn13Stochastic,Mairal14incremental}.
\end{remark}

%

\section{Convergence Analysis}\label{sec:convergence}
In order to reduce the notational burden, our analysis will be based on Algorithm 1, which uses a global clock.  We first make a few assumptions.

\pn {\bf Assumption A.}
\begin{itemize}
\item[A1.] (On the Problem) There exists a positive constant $L_k>0$ such that $$\|\nabla_k g_k(x_k)-\nabla_k g_k(z_k)\|\le L_k \|x_k-z_k\|, \; \forall~x_k,z_k, \; \forall~k.$$
Moreover, $h$ is convex (possibly nonsmooth); $X$ is a closed, convex and compact set. $f(x)$ is bounded from below over $X$.
\item[A2.] (On the Asynchrony) The total delays are bounded, i.e., for each node $k$ there exists finite constants $T_k$ such that $t-t(k)\le T_k$ for all $t$ and $k$.
\item[A3.] (On the Algorithm)
For all $k$, the stepsize $\rho_k$ is chosen large enough such that:
\begin{align}
\alpha_k&:={\rho_k}-2\left(\frac{1}{\rho_k}+\frac{7L_k}{2\rho_k^2} \right) L^2_k (T_k+1)^2-L_k T^2_k>0\label{eq:alpha}\\
\rho_k& > 7 L_k, \; k=1,\cdots, K.
\end{align}
\end{itemize}
By Assumption A2, we see that the only requirement on the asynchrony is that when each $x_k$ is updated, the information used to compute the gradient should be one of $x$ generated within last $T_k$ iterations. So it is perfectly legitimate if copies of $x$ or copies of the gradients get lost due to unsuccessful communication. Also there is nothing preventing copies of $x$ from arriving at the same node $k$ with reversed order (e.g., $x^{t}$ arrives after $x^{t+1}$). Due to this assumption on the boundedness of the asynchrony, Algorithms 1 belongs to the family of ``partially asynchronous algorithm", as opposed to the ``totally asynchronous algorithm" in which the delays can potentially be unbounded \footnote{In short, the only requirement for the totally asynchronous algorithm is that no nodes quits forever. }; see the definitions and discussions in \cite{Bertsekas_Book_Distr}.

From Assumption A3, it is clear that when the system is synchronous, i.e., when $T_k= 0 $, the bound for $\alpha_k$ becomes
\begin{align}
\alpha_k&:={\rho_k}-\left(\frac{1}{\rho_k}+\frac{7L_k}{2\rho_k^2} \right)2 L^2_k>0\label{eq:beta:2}.
\end{align}

We have the following result.
\begin{lemma}\label{lemma:y2}
Suppose Assumption A is satisfied. Then for Algorithm 1, the following is true for all $k$
\begin{align}
\begin{split}
&\|y^{t+1}_k-y^{t}_k\|^2\le L^2_k (T_k+1)\sum_{i=0}^{T_k}\|x^{t+1-i}-x^{t-i}\|^2.\label{eq:y_difference2}
\end{split}
\end{align}
\end{lemma}
\begin{proof}
From the update of $x_k$ in  \eqref{eq:x_k_update_prox}, we observe that the following is true
\begin{align}
\nabla g_k(x^{[t+1](k)})+y^{t}_k+\rho_k(x^{t+1}_k-x^{t+1})=0,
\end{align}
or equivalently
\begin{align}\label{eq:y_expression2}
\nabla g_k(x^{[t+1](k)})=-y^{t+1}_k.
\end{align}
Note that both $x_k$ and $y_k$ are updated at each iteration, so we have the following equality for iteration $t$ as well
\begin{align}\label{eq:y_expression3}
\nabla g_k(x^{[t](k)})=-y^{t}_k.
\end{align}
Suppose $k\notin \cC^{t+1}$, which means that no new gradient information arrives for node $k$. In this case, we have  $[t+1](k) = t(k)$, therefore
\begin{align}
\nabla g_k(x^{[t+1](k)}) = \nabla g_k(x^{[t](k)}),\; \forall~k\notin \cC^{t+1}.
\end{align}
This combined with \eqref{eq:y_expression2} -- \eqref{eq:y_expression3} yields
\begin{align}
\|y_k^{t+1} - y^t_k\|^2 = 0, \; \forall~k\notin\cC^{t+1}.
\end{align}
It follows that for $k\notin\cC^{t+1}$, \eqref{eq:y_difference2} is true.

Suppose that $k\in\cC^{t+1}$, then we have
\begin{align*}
&|[t+1](k)-[t](k)|\nonumber\\
&\le
\left\{ \begin{array}{ll}t+1-[t](k)\le T_k+1, & \mbox{if}\; [t+1](k)\ge [t](k),\\
t+1-[t+1](k)\le T_k, &\mbox{otherwise}
\end{array}\right..
\end{align*}

Therefore we have, for all $k\in\cC^{t+1}${\small
\begin{align}
\|y^{t+1}_k-y^{t}_k\|&= \|\nabla g_k(x^{[t+1](k)})-\nabla g_k(x^{[t](k)})\|\nonumber\\
&\le L_k\sum_{i=0}^{T_k}\|x^{t+1-i}-x^{t-i}\|.
\end{align}}
The above result further implies that,{\small
\begin{align}
&\|y^{t+1}_k-y^{t}_k\|^2\le L^2_k (T_k+1)\sum_{i=0}^{T_k}\|x^{t+1-i}-x^{t-i}\|^2, \; k\in\cC^{t+1}.\nonumber
\end{align}}
The desired result is obtained.
\end{proof}

Next, we {upper} bound the successive difference of the augmented Lagrangian. To this end, let us define a few new functions, given below{\small
\begin{subequations}\label{eq:defL}
\begin{align}
\ell_k(x_k; x^{t+1},y^t)&=g_k(x_k)+\langle y_k^{t}, x_k-x^{t+1}\rangle+\frac{\rho_k}{2}\|x_k-x^{t+1}\|^2\\
u_k(x_k; x^{t+1},y^t)&=g_k(x^{t+1})+\langle \nabla g_k(x^{t+1}), x_k-x^{t+1}\rangle+\langle y_k^t, x_k-x^{t+1}\rangle+\frac{\rho_k}{2}\|x_k-x^{t+1}\|^2,\\
\bar{u}_k(x_k; x^{t+1},y^t)&=g_k(x^{t+1})+\langle \nabla g_k(x^{[t+1](k)}), x_k-x^{t+1}\rangle+\langle y_k^t, x_k-x^{t+1}\rangle+\frac{\rho_k}{2}\|x_k-x^{t+1}\|^2.
\end{align}
\end{subequations}}

Using these short-handed definitions, we have
\begin{align}
&L(\{x^{t+1}_k\}, x^{t+1}; y^{t})=\sum_{k=1}^{K}\ell_k(x^{t+1}_k; x^{t+1}, y^t)\label{eq:L_ell}\\
&x^{t+1}_k=\arg\min_{x_k} \bar{u}_k(x_k;x^{t+1},y^t).\label{eq:x_k_u_k}
\end{align}

The lemma below bounds the difference between $\ell_k(x^{t+1}_k; x^{t+1},y^t)$ and $\ell_k(x_{k}^{t}; x^{t+1},y^t)$.
\begin{lemma}\label{lemma:l_difference}
Suppose Assumption A is satisfied. Let $\{x^{t}_k,x^{t},y^t\}$ be generated by Algorithm 1. Then we have the following
\begin{align}\label{eq:bound_l}
&\ell_k(x^{t+1}_k;x^{t+1},y^t)-\ell_k(x^{t}_k;x^{t+1},y^t)\nonumber\\
&\le -\left(\frac{\rho_k}{2}-\frac{7}{2}L_k\right)\|x_k^{t}-x_k^{t+1}\|^2+\frac{L_k T_k}{2}\sum_{i=0}^{T_k-1}\|x^{t+1-i}-x^{t-i}\|^2\nonumber\\
&\quad+\frac{7L_k}{2\rho_k^2} \|y_k^{t+1}-y_k^{t}\|^2, \quad k=1,\cdots,K.
\end{align}
\end{lemma}
\begin{proof}
From the definition of $\ell_k(\cdot)$ and $u_k(\cdot)$ we have the following
\begin{align}\label{eq:LU}
\ell_k(x_k;x^{t+1},y^t)\le u_k(x_k;x^{t+1},y^t) +\frac{L_k}{2}\|x_k-x^{t+1}\|^2, \; \forall~x_k,\; k=1,\cdots,K.
\end{align}
Observe that $x^{t+1}_k$ is generated according to \eqref{eq:x_k_u_k}. Combined with the strong convexity of $\bar{u}_k(x_k;x^{t+1},y^t)$ with respect to $x_k$, we have
\begin{align}
&\bar{u}_k(x^{t+1}_k;x^{t+1},y^t)-\bar{u}_k(x^{t}_{k};x^{t+1},y^t)\nonumber\\
&\le -\frac{\rho_k}{2}\|x^{t}_{k}-x^{t+1}_k\|^2, \; \forall~k, \label{eq:u_k_bar_difference}\\
&\nabla \bar{u}_k(x^{t+1}_k;x^{t+1},y^t) = 0. \label{eq:u_k_gradient_0}
\end{align}

Also we have
\begin{align}
&\nabla \bar{u}_k(x^{t+1}_k;x^{t+1},y^t)-\nabla {u}_k(x^{t+1}_k;x^{t+1},y^t)\nonumber\\
&=\nabla g_k(x^{[t+1](k)}) - \nabla g_k(x^{t+1}) \label{eq:difference_gradient}.
\end{align}

Using the strong convexity of $u_k$, we have the series of inequalities given below
\begin{align}
u_k(x^{t+1}_k; x^{t+1}, y^t)
&\le u_k(x^t_k; x^{t+1}, y^t)+\langle\nabla u_k(x^{t+1}_k; x^{t+1}, y^t), x^{t+1}_k-x^t_k\rangle-\frac{\rho_k}{2}\|x^{t+1}_k-x^t_k\|^2\nonumber\\
&= u_k(x^t_k; x^{t+1}, y^t)+ \langle\nabla u_k(x^{t+1}_k; x^{t+1}, y^t)-\nabla \bar{u}_k(x^{t+1}_k; x^{t+1}, y^t), x^{t+1}_k-x^t_k\rangle\nonumber\\
&\quad+\langle\nabla \bar{u}_k(x^{t+1}_k; x^{t+1}, y^t), x^{t+1}_k-x^t_k\rangle-\frac{\rho_k}{2}\|x^{t+1}_k-x^t_k\|^2\nonumber\\
&= u_k(x^t_k; x^{t+1}, y^t)+ \langle\nabla g_k(x^{t+1})-\nabla g_k(x^{[t+1](k)}), x^{t+1}_k-x^t_k\rangle-\frac{\rho_k}{2}\|x^{t+1}_k-x^t_k\|^2\nonumber\\
&\le u_k(x^t_k; x^{t+1}, y^t)+ L_k\|x^{[t+1](k)}-x^{t+1}\|\|x_k^{t+1}-x^t_k\|-\frac{\rho_k}{2}\|x^{t+1}_k-x^t_k\|^2\nonumber\\
&\le u_k(x^t_k; x^{t+1}, y^t)+ \frac{L_k T_k}{2}\sum_{i=0}^{T_k-1}\|x^{t+1-i}-x^{t-i}\|^2-\frac{\rho_k-L_k}{2}\|x^{t+1}_k-x^t_k\|^2.\label{eq:u_k_difference}
\end{align}

Further, we have the following series of inequalities
\begin{align}\label{eq:u_k_l_k}
&{u}_k(x^{t}_{k};x^{t+1},y^t)-\ell_k(x^{t}_{k};x^{t+1},y^t)\nonumber\\
&=g_k(x^{t+1})+\langle \nabla g_k(x^{t+1}), x^{t}_{k}-x^{t+1}\rangle\nonumber\\
&\quad \quad +\langle y_k^t, x^{t}_{k}-x^{t+1}\rangle+\frac{\rho_k}{2}\|x^{t}_{k}-x^{t+1}\|^2\nonumber\\
&\quad \quad -\left(g_k(x^{t}_{k})+\langle y^{t}_{k}, x^{t}_{k}-x^{t+1}\rangle+\frac{\rho_k}{2}\|x^{t}_{k}-x^{t+1}\|^2\right)\nonumber\\
&=g_k(x^{t+1})-g_k(x^{t}_{k})+\langle \nabla g_k(x^{t+1}), x^{t}_{k}-x^{t+1}\rangle\nonumber\\
&\le\langle \nabla g_k(x^{t+1})-\nabla g_k(x_k^{t}), x^{t}_{k}-x^{t+1}\rangle+\frac{L_k}{2}\|x^{t}_{k}-x^{t+1}\|^2\nonumber\\
&\le \frac{3}{2}L_k\|x^{t}_{k}-x^{t+1}\|^2\nonumber\\
&\le 3L_k\left(\|x^{t}_{k}-x^{t+1}_k\|^2+\|x^{t+1}_k-x^{t+1}\|^2\right),
\end{align}
where the first two inequalities follow from Assumption A1.
Combining {\eqref{eq:LU} -- \eqref{eq:u_k_l_k}} 
we obtain
\begin{align}
&\ell_k(x^{t+1}_k; x^{t+1},y^t)-\ell_k(x^{t}_k; x^{t+1},y^t)\nonumber\\
&\le u_k(x^{t+1}_k; x^{t+1},y^t)-u_k(x_k^{t}; x^{t+1},y^t)+\frac{L_k}{2}\|x^{t+1}-x^{t+1}_k\|^2\nonumber\\
&\quad\quad+u_k(x_k^{t}; x^{t+1},y^t)-\ell_k(x^{t}_k; x^{t+1},y^t)\nonumber\\
&\le-\frac{\rho_k-L_k}{2}\|x_k^{t}-x_k^{t+1}\|^2+\frac{L_k T_k}{2}\sum_{i=0}^{T_k-1}\|x^{t+1-i}-x^{t-i}\|^2\nonumber\\
&\quad +\frac{7L_k}{2\rho_k^2} \|y_k^{t+1}-y_k^{t}\|^2+3L_k\|x^{t}_k-x_k^{t+1}\|^2\nonumber.
\end{align}
The desired result then follows.
\end{proof}

Next, we bound the difference of the augmented Lagrangian function values.
\begin{lemma}\label{lemma:L_difference2}
Assume the same set up as in Lemma~\ref{lemma:l_difference}. Then we have
\begin{align}\label{eq:L_descent2}
&L(\{x^{t+1}_k\}, x^{t+1}; y^{t+1})-L(\{x^{1}_k\}, x^{1}; y^{1})\nonumber\\
&\le-\sum_{i=1}^{t}\sum_{k=1}^{K}\left(\frac{\rho_k-7L_k}{2}\right)\|x^{i+1}_k-x^i_k\|^2-\sum_{i=1}^{t}\sum_{k=1}^{K}\alpha_k\|x^{i+1}-x^i\|^2
\end{align}
where $\alpha_k$ is the constant defined in \eqref{eq:alpha}.
\end{lemma}
\begin{proof}
We first bound the successive difference $L(\{x^{t+1}_k\}, x^{t+1}; y^{t+1})-L(\{x^{t}_k\}, x^{t}; y^{t})$. We first decompose the difference by
\begin{align}
&L(\{x^{t+1}_k\}, x^{t+1}; y^{t+1})-L(\{x^{t}_k\}, x^{t}; y^{t})\nonumber\\
&=\left(L(\{x^{t+1}_k\}, x^{t+1}; y^{t+1})-L(\{x^{t+1}_k\}, x^{t+1}; y^{t})\right) \nonumber\\
&\quad\quad+ \left(L(\{x^{t+1}_k\}, x^{t+1}; y^{t})-L(\{x^{t}_k\}, x^{t}; y^{t})\right).\label{eq:successive_L1}
\end{align}

The first term in \eqref{eq:successive_L1} can be expressed as
\begin{align}
&L(\{x^{t+1}_k\}, x^{t+1}; y^{t+1})-L(\{x^{t+1}_k\}, x^{t+1}; y^{t})\nonumber\\
&=\sum_{k=1}^{K}\frac{1}{\rho_k}\|y_k^{t+1}-y_k^{t}\|^2\nonumber.
\end{align}
To bound the second term in \eqref{eq:successive_L1}, we use Lemma \ref{lemma:l_difference}. We have the series of inequalities in \eqref{eq:LfirstBound}, where the last inequality follows from Lemma~\ref{lemma:l_difference} and the strong convexity of $L(\{x^{t}_k\}, x; y^{t})$ with respect to the variable $x$ (with modulus $\gamma=\sum_{k=1}^{K}\rho_k$)  at $x=x^{t+1}$.{\small
\begin{align}
&L(\{x^{t+1}_k\}, x^{t+1}; y^{t})-L(\{x^{t}_k\}, x^{t}; y^{t})\nonumber\\
&=L(\{x^{t+1}_k\}, x^{t+1}; y^{t})-L(\{x^{t}_k\}, x^{t+1}; y^{t})+L(\{x^{t}_k\}, x^{t+1}; y^{t})-L(\{x^{t}_k\}, x^{t}; y^{t})\nonumber\\
&= \sum_{k=1}^{K}\left(\ell_k(x^{t+1}_k;x^{t+1},y^{t})-\ell_k(x^{t}_k;x^{t+1},y^{t})\right)+L(\{x^{t}_k\}, x^{t+1}; y^{t})-L(\{x^{t}_k\}, x^{t}; y^{t})\nonumber\\
&\le -\sum_{k=1}^{K}\bigg[\left(\frac{\rho_k}{2}-\frac{7}{2}L_k\right)\|x_k^{t}-x_k^{t+1}\|^2 - \frac{L_k T_k }{2}\sum_{i=0}^{T_k-1}\|x^{t+1-i}-x^{t-i}\|^2- \frac{7L_k}{2\rho_k^2} \|y_k^{t+1}-y_k^{t}\|^2 \bigg]-\frac{1}{2}\sum_{k=1}^{K}\rho_k\|x^{t+1}-x^t\|^2 \label{eq:LfirstBound}
\end{align}}

Combining the above two inequalities and use Lemma \ref{lemma:y2}, we obtain the inequality below:
\begin{align}\label{eq:L_difference2}
&L(\{x^{t+1}_k\}, x^{t+1}; y^{t+1})-L(\{x^{t}_k\}, x^{t}; y^{t})\le \sum_{k=1}^{K}\bigg[-\left(\frac{\rho_k}{2}-\frac{7}{2}L_k\right)\|x_k^{t}-x_k^{t+1}\|^2 \nonumber\\
&\quad +\frac{L_kT_k}{2}\sum_{i=0}^{T_k-1}\|x^{t+1-i}-x^{t-i}\|^2\bigg]-\frac{1}{2}\sum_{k=1}^{K}\rho_k\|x^{t+1}-x^t\|^2\nonumber\\ &\quad +\sum_{k=1}^{K}\left(\frac{1}{\rho_k}+\frac{7L_k}{2\rho_k^2} \right)\left(L^2_k (T_k+1)\sum_{i=0}^{T_k}\|x^{t+1-i}-x^{t-i}\|^2\right).
\end{align}

Then for any given $t$, the difference $L(\{x^{t+1}_k\}, x^{t+1}; y^{t+1})-L(\{x^{1}_k\}, x^{1}; y^{1})$ is obtained by summing \eqref{eq:L_difference2} {over all} iterations:
\begin{align}\label{eq:L_difference3}
&L(\{x^{t+1}_k\}, x^{t+1}; y^{t+1})-L(\{x^{1}_k\}, x^{1}; y^{1})\nonumber\\
&\le -\sum_{i=1}^{t}\sum_{k=1}^{K}\left(\frac{\rho_k}{2}-\frac{7}{2} L_k\right)\|x^{i+1}_k-x^i_k\|^2
\nonumber\\
& \quad -\sum_{i=1}^{t}\sum_{k=1}^{K}\left(\frac{\rho_k}{2}-\left(\frac{1}{\rho_k}+\frac{7L_k}{2\rho_k^2} \right) L^2_k (T_k+1)^2-\frac{L_k T^2_k}{2}\right)\|x^{i+1}-x^i\|^2\nonumber\\
&:= -\sum_{i=1}^{t}\sum_{k=1}^{K}\frac{\rho_k-7L_k}{2} \|x^{i+1}_k-x^i_k\|^2
-\sum_{i=1}^{t}\sum_{k=1}^{K}\alpha_k \|x^{i+1}-x^i\|^2.
\end{align}
This completes the proof.
\end{proof}

We conclude that to make the augmented Lagrangian decrease at each iteration, it is sufficient to require that  $\alpha_k>0$ and $\rho_k-7L_k>0$ for all $k$. Note that one can always find a set of $\rho_k$'s large enough such that the above condition is satisfied. 

Next we show that $L(\{x^{t}_k\}, x^{t}; y^{t})$ is convergent.
\begin{lemma}\label{lemma:L_bounded2}
Suppose Assumption A is satisfied. Then Algorithm 1 generates a sequence of augmented Lagrangian that satisfies
\begin{align}
\lim_{t\to\infty}L(\{x^{t}_k\}, x^t, y^t)\ge \inf_{x\in X} f(x) - \mbox{diam}^2(X)\sum_{k=1}^{K}\frac{L_k}{2} >-\infty
\end{align}
where $\mbox{diam}(X):=\sup\{\|x_1-x_2\|\mid x_1, x_2\in X\}$, which is the diameter of the set $X$.
\end{lemma}
\begin{proof}
Observe that the augmented Lagrangian can be expressed as{\small
\begin{align}\label{eq:L_lower_bound2}
&L(\{x^{t+1}_k\}, x^{t+1}; y^{t+1})\nonumber\\
&=h(x^{t+1})+\sum_{k=1}^{K}\left(g_k(x^{t+1}_k)+\langle y^{t+1}_k, x^{t+1}_k-x^{t+1}\rangle+\frac{\rho_k}{2}\|x^{t+1}_k-x^{t+1}\|^2\right)\nonumber\\
&\stackrel{\rm (a)}=h(x^{t+1})+\sum_{k=1}^{K}\bigg(g_k(x^{t+1}_k)+\langle \nabla g_k(x^{[t+1](k)}), x^{t+1}-x_k^{t+1}\rangle\nonumber\\
&\quad +\frac{\rho_k}{2}\|x^{t+1}_k-x^{t+1}\|^2\bigg)\nonumber\\
&=h(x^{t+1})+\sum_{k=1}^{K}\bigg(g_k(x^{t+1}_k)+\langle \nabla g_k(x^{[t+1](k)})-\nabla g_k(x^{t+1}), x^{t+1}-x_k^{t+1}\rangle\nonumber\\
&\quad+\langle \nabla g_k(x^{t+1}), x^{t+1}-x_k^{t+1}\rangle+\frac{\rho_k}{2}\|x^{t+1}_k-x^{t+1}\|^2\bigg)\nonumber\\
&\stackrel{\rm (b)}\ge h(x^{t+1})+\sum_{k=1}^{K}\bigg(g_k(x^{t+1})+\frac{\rho_k-3L_k}{2}\|x^{t+1}_k-x^{t+1}\|^2\nonumber\\
&\quad-L_k\|x^{[t+1](k)}-x^{t+1}\|\|x^{t+1}-x^{t+1}_k\|\bigg)\nonumber\\
&\ge f(x^{t+1})+\sum_{k=1}^{K}\left(\frac{\rho_k-4L_k}{2}\|x^{t+1}_k-x^{t+1}\|^2-\frac{L_k}{2}\|x^{[t+1](k)}-x^{t+1}\|^2\right)\nonumber\\
&\stackrel{(c)}\ge \inf_{x\in X} f(x) - \mbox{diam}^2(X)\sum_{k=1}^{K}\frac{L_k}{2}\stackrel{(d)}\ge -\infty.
\end{align}}
In the above series of inequalities, $\rm (a)$ is from \eqref{eq:y_expression2}; $\rm (b)$ is due to the Cauchy-Schwartz inequality, Assumption A1, and the following inequalities{\small
\begin{align}
g_k(x^{t+1})&\le g_k(x^{t+1}_k)+\langle\nabla g_k(x^{t+1}_k), x^{t+1}-x^{t+1}_k\rangle+\frac{L_k}{2}\|x^{t+1}_k-x^{t+1}\|^2\nonumber\\
&= g_k(x^{t+1}_k)+\langle\nabla g_k(x^{t+1}_k)-\nabla g_k(x^{t+1}), x^{t+1}-x^{t+1}_k\rangle\nonumber\\
&\quad+\langle\nabla g_k(x^{t+1}), x^{t+1}-x^{t+1}_k\rangle+\frac{L_k}{2}\|x^{t+1}_k-x^{t+1}\|^2\nonumber\\
&\le g_k(x^{t+1}_k)+\langle\nabla g_k(x^{t+1}), x^{t+1}-x^{t+1}_k\rangle+\frac{3L_k}{2}\|x^{t+1}_k-x^{t+1}\|^2.\nonumber
\end{align}}
The inequality in $(c)$ is due to the assumption that $\rho_k\ge 4L_k$, and by the definition of the $\mbox{diam}(X)$; $(d)$ is because
of the assumption that $f(x)$ is bounded over all $X$, and that $X$ is a compact set. It follows from Lemma \ref{lemma:L_difference2} that whenever the stepsize $\rho_k$'s are chosen sufficiently large (as per Assumption A), $L(\{x^{t+1}_k\}, x^{t+1}; y^{t+1})$ will monotonically decrease and is convergent. This completes the proof.
\end{proof}

Using Lemmas \ref{lemma:y2}--\ref{lemma:L_bounded2}, we arrive at the following convergence result.
\begin{theorem}\label{thm:convergence2}
{Suppose that Assumption A holds. Then the following is true for Algorithm 1.
\begin{enumerate}
\item We have $\lim_{t\to\infty}\|x^{t+1}-x^{t+1}_k\|=0$, $k=1,\cdots, K$. That is, the primal feasibility is always satisfied in the limit.
\item The sequence $\{\{x^{t+1}_k\}, x^{t+1}, y^{t+1}\}$ converges to the set of stationary solutions of problem \eqref{eq:consensus:admm}. Moreover, the sequence $\{\{x^{t+1}_k\}, x^{t+1}\}$  converges to the set of stationary solutions of problem \eqref{eq:consensus}.
\end{enumerate}}
\end{theorem}

\begin{proof}
Combining Lemma \ref{lemma:L_difference2}  -- \ref{lemma:L_bounded2} we must have
\begin{align}
&\lim_{t\to\infty}\|x^{t+1}_k-x^t_k\| \to 0, \; \forall~k,\nonumber\\
&\lim_{t\to\infty}\|x^{t+1}-x^t\| \to 0\nonumber.
\end{align}
Taking limit on both sides of \eqref{eq:y_difference2} and use the above two results, we immediately obtain
\begin{align}
\lim_{t\to\infty}\|y^{t+1}_k-y^t_k\|\to 0,\; \forall ~k.
\end{align}
The first part of the claim is proven.

Once we can show that the primal feasibility gap goes to zero, the proof for stationarity is straightforward. We refer the readers to \cite{hong14nonconvex_admm} for detailed arguments.
\end{proof}

It turns out that for some special cases of $g_k$'s, the requirement on the stepsize can be further relaxed.
\begin{corollary}\label{cor:l_difference}
Suppose Assumption A1 and A3 are true. We have the following:

\begin{enumerate}
\item If $g_k$ is a convex function, then the corresponding $\rho_k$ should satisfy:
\begin{align}\label{eq:rho_convex}
\begin{split}
&{\rho_k}-2\left(\frac{1}{\rho_k}+\frac{L_k}{2\rho_k^2} \right) L^2_k (T_k+1)^2-L_k T^2_k>0\\
&\rho_k \ge  L_k, \; k=1,\cdots, K.
\end{split}
\end{align}


\item If $g_k$ is a concave function, then the corresponding $\rho_k$ should satisfy:
\begin{align}\label{eq:rho_concave}
\begin{split}
&{\rho_k}-2\left(\frac{1}{\rho_k}+\frac{5L_k}{2\rho_k^2} \right) L^2_k (T_k+1)^2-L_k T^2_k>0\\
&\rho_k \ge 5 L_k, \; k=1,\cdots, K.
\end{split}
\end{align}
\end{enumerate}
\end{corollary}
\begin{proof}
The proof is similar to that leading to Theorem \ref{thm:convergence2}. For the convex case, the only differences is in bounding the size of the difference between ${u}_k(x^{t}_{k};x^{t+1},y^t)$ and $\ell_k(x^{t}_{k};x^{t+1},y^t)$ in Lemma \ref{lemma:l_difference}, and in Lemma \ref{lemma:L_bounded2}. If $g_k$  is convex, we have
\begin{align}\label{eq:u_k_l_k2}
&{u}_k(x^{t}_{k};x^{t+1},y^t)-\ell_k(x^{t}_{k};x^{t+1},y^t)\nonumber\\
&=g_k(x^{t+1})-g_k(x^{t}_{k})+\langle \nabla g_k(x^{t+1}), x^{t}_{k}-x^{t+1}\rangle\nonumber\\
&\le 0
\end{align}
where the last inequality comes from the convexity of $g_k$. Similarly, we can replace the series of inequalities in \eqref{eq:L_lower_bound2} by
\begin{align}
&L(\{x^{t+1}_k\}, x^{t+1}; y^{t+1})\nonumber\\
&=h(x^{t+1})+\sum_{k=1}^{K}\bigg(g_k(x^{t+1}_k)+\langle \nabla g_k(x^{[t+1](k)})-\nabla g_k(x^{t+1}), x^{t+1}-x_k^{t+1}\rangle\nonumber\\
&\quad+\langle \nabla g_k(x^{t+1}), x^{t+1}-x_k^{t+1}\rangle+\frac{\rho_k}{2}\|x^{t+1}_k-x^{t+1}\|^2\bigg)\nonumber\\
&\stackrel{\rm (a)}\ge h(x^{t+1})+\sum_{k=1}^{K}\bigg(g_k(x^{t+1})+\frac{\rho_k}{2}\|x^{t+1}_k-x^{t+1}\|^2\nonumber\\
&\quad-L_k\|x^{[t+1](k)}-x^{t+1}\|\|x^{t+1}-x^{t+1}_k\|\bigg)\nonumber\\
&\ge f(x^{t+1})+\sum_{k=1}^{K}\left(\frac{\rho_k-L_k}{2}\|x^{t+1}_k-x^{t+1}\|^2-\frac{L_k}{2}\|x^{[t+1](k)}-x^{t+1}\|^2\right)
\end{align}
where in $(a)$ we have again used the Cauchy-Schwartz inequality and the convexity of $g_k$. Then by simple manipulation we arrive at the claimed result.

If $g_k$  is concave, we can bound the difference between ${u}_k(x^{t}_{k};x^{t+1},y^t)$ and $\ell_k(x^{t}_{k};x^{t+1},y^t)$ in Lemma \ref{lemma:l_difference} by
\begin{align}\label{eq:u_k_l_k2}
&{u}_k(x^{t}_{k};x^{t+1},y^t)-\ell_k(x^{t}_{k};x^{t+1},y^t)\nonumber\\
&=g_k(x^{t+1})-g_k(x^{t}_{k})+\langle \nabla g_k(x^{t+1}), x^{t}_{k}-x^{t+1}\rangle\nonumber\\
&\stackrel{(a)}\le\langle \nabla g_k(x^{t+1})-\nabla g_k(x_k^{t}), x^{t}_{k}-x^{t+1}\rangle\nonumber\\
&\le 2L_k\left(\|x^{t}_{k}-x^{t+1}_k\|^2+\|x^{t+1}_k-x^{t+1}\|^2\right),
\end{align}
where $(a)$ comes from the concavity of $g_k$.
\end{proof}

We have a few remarks in order.

\begin{remark}
{\it(On the Bounded Delays)}
Our convergence results critically dependent on the choice of the stepsizes $\{\rho_k\}$, which in turn is a function of the bounds $\{T_k\}$. Clearly all $T_k$'s must be finite, therefore the scheme proposed here is reminiscent to the family of {\it partially asynchronous} algorithm discussed in \cite[Chapter 7]{Bertsekas_Book_Distr}. Clearly those bounds on $\rho_k$'s are developed for the {\it worst case} delay scenarios. If we model different delays $\{t-t(k)\}$ as random variables following certain probability distributions with {\it finite} supports, we can slightly modify the analysis so that the final bounds on $\rho_k$'s are dependent on the statistical properties of the random variables. Such modification is minor so we do not intend to go over it in this paper. The more interesting case would be when the delays $\{t-t(k)\}$ follow distributions with finite means and variances but {\it unbounded} supports. However our current approach cannot be directly used.
\end{remark}

\begin{remark}
{\it(On the Relationship with \cite{hong14nonconvex_admm})}
The analysis presented above follows the general {recipe} first alluded in \cite{Ames13LDA} and later generalized in \cite{hong14nonconvex_admm}, for dealing with nonconvex ADMM-type algorithms. The same three-step approach is used here: {\it 1)} Bounding the size of the successive difference of the dual variable; {\it 2)} Bounding the successive difference of the augmented Lagrangian; {\it 3)} Bounding the sequence of the augmented Lagrangian. However, several important improvements have been made to both the algorithm and the analysis in order to better incorporate asynchrony. For example, compared with the flexible Proximal ADMM algorithm in \cite{hong14nonconvex_admm}, we have increased the stepsize for updating $x_k$ from $\frac{1}{\rho_k+L_k}$ to $\frac{1}{\rho_k}$. This change significantly simplifies the analysis and leads to a better bound for $\rho_k$.  Second, in the flexible Proximal ADMM, a given tuple $(x_k, y_k)$ is only updated when the new gradient is available, while here $(x_k, y_k)$ is updated at every iteration regardless of the availability of new gradients. This also leads to better bound for $\rho_k$ and faster algorithm. Third, different analysis techniques have been used throughout to take into consideration the changes in the algorithm as well as the presence of staled gradients.
\end{remark}

\begin{remark}
{\it (On the Necessity of Lemma \ref{lemma:L_bounded2})}
As a technical remark, we emphasize that lower-bounding the sequence of the augmented Lagrangian, as we have done in Lemma \ref{lemma:L_bounded2}, is a key step in the entire analysis. The reason is that the compactness of the set $X$ only guarantees that the sequence of $\{x^t\}$ is bounded, but not the sequences $\{x^t_k, y^t_k\}$ (note that $x^t_k$ is generated by solving an {\it unconstrained} problem). Without the boundedness of $\{x^t_k, y^t_k\}$, the augmented Lagrangian $L(\{x^t_k\}, x^t; y^t)$ can go to $-\infty$, therefore one cannot claim that $\|x^{t+1}-x^t\|\to 0$ and $\|x^t_k-x^{t+1}_k\|\to 0$.
\end{remark}

\section{Numerical Results}\label{sec:numerical}
In this section we conduct numerical experiments to validate the performance of the proposed algorithm.

\subsection{The Setup}
We consider the following nonconvex problem:
\begin{align}
\begin{split}
\min&\quad f(x):=-\frac{1}{2}\sum_{k=1}^{K}x^T B^T_k B_k x+\lambda\|x\|_1\\
\st &\quad \|x\|^2_2\le 1 \label{eq:example},
\end{split}
\end{align}
where $B_k\in \mathbb{R}^{M_k
\times N}$ is a data matrix and $\lambda\ge 0$ is some constant. When $K=1$, this problem is related to the $\ell_1$ penalized version of the sparse principal component analysis (PCA) problem; see \cite{richtarik12PCA, hong14icassp,jiang14complexity}. This problem is also related to the penalized version of Fisher's Linear Discriminant Analysis (LDA); see \cite{Witten11,Ames13LDA}. However, the algorithms discussed in the existing literature such as those in \cite{richtarik12PCA, hong14icassp,jiang14complexity, Witten11,Ames13LDA} cannot deal with the scenario where the data matrices $\{B_k\}_k$ are physically located in distributed nodes.

To formulate above sparse PCA problem in the form of \eqref{eq:consensus:admm}, we introduce a set of new variable $\{x_k\}$:
\begin{align}
\begin{split}
\min&\quad -\frac{1}{2}\sum_{k=1}^{K}x_k^T B^T_k B_k x_k+\lambda\|x\|_1\\
\st &\quad \|x\|^2\le 1, \; x_k = x,\; \forall~k. \label{eq:example:consensus}
\end{split}
\end{align}
It is straightforward to see that when applying ADMM or the Async-PADMM, each subproblem can be solved in closed form. It is also worth noting that each smooth term in the objective of \eqref{eq:example:consensus}, $-\frac{1}{2}x_k^T B^T_k B_k x_k $, is a concave function, so the refined the stepsize rule \eqref{eq:rho_concave} can be used for Async-PADMM.

In our experiment, we compare the Async-PADMM with the following two algorithms
\begin{enumerate}
\item {\bf Synchronous ADMM Algorithm}: This is the vanilla ADMM algorithm discussed in \cite[Section 2.2]{hong14nonconvex_admm}. The algorithm can handle nonconvex problems, but its protocol is synchronous. Therefore when the nodes have different computational time, the master node has to wait for all the distributed nodes to complete one iteration of computation before proceeding to the next step. The downside of this approach is that fast nodes have to wait for the slow nodes. The choice of the stepsize $\rho_k$ follows the condition given in  \cite[Assumption A]{hong14nonconvex_admm}.
\item {\bf Synchronous PADMM Algorithm}: This is the period-1 proximal ADMM algorithm discussed in \cite[Section 2.3]{hong14nonconvex_admm}. Again the algorithm is synchronous. The choice of the stepsize $\rho_k$ follows the condition given in  \cite[Assumption B]{hong14nonconvex_admm}.
\end{enumerate}
It is worth noting that by simply waiting for the slowest nodes at each iteration, the ADMM and PADMM are capable of handling the asynchrony caused by the computational delay, albeit in a rather inefficient way. However neither of them can deal with the asynchrony caused by imperfect communication link (i.e., loss of messages, out-of-sequence messages, etc). Therefore for fair comparison, in our experiments we only consider scenarios where the communication links are perfect. That is, we require all messages sent by the nodes are perfectly received,  in sequence.

In our experiment, the data matrices $\{B_k\}_{k}$ are generated as follows. Each element $b_k(i,j), i=1, \cdots M_k, j=1, \cdots, N$ in $B_k\in\mathbb{R}^{M_k\times N}$ is generated independently, according to the following Gaussian Mixture model: $$b_k(i,j)\sim p_k \mathcal{N}(a_k(i,j), c_k(i,j))+(1-p_k)\mathcal{N}(0, 0),$$
where $a_k(i,j)$ and $c_k(i,j)$ follow uniform distribution: $a_k(i,j)\sim \mbox{Uniform}(0,1)$ and $c_k(i,j)\sim \mbox{Uniform}(0,1)$.  In words, there is a probability $p_k>0$ such that $b_k(i,j)$ is nonzero and follows a Gaussian distribution with random mean and variance. The asynchrony in the system is simulated as follows. For each node $k$ we assign a distinct $T_k$ which characterizes the maximum delay for that node. Each time node $k$ starts to perform computation (i.e., Step S2 in Algorithm 1(b)), the computational delay is randomly drawn from the distribution $\mbox{Uniform}(0,T_k)$.

To measure the progress of different algorithms, we need the following definitions. For a given iterate $x^{t}$, it is known that the size of the proximal gradient, expressed below, can be used to measure the optimality:
\begin{align}\label{eq:proximal_gradient}
\begin{split}
\tilde{\nabla} f(x^{t})& = x^{t} - \prox_{h+\iota(X)}\left(x^t - (f(x^{t})-h(x^{t}))\right)\\
&= x^{t} - \prox_{h+\iota(X)}\left(x^{t} + (x^t)^T\sum_{k=1}^{K} B^T_k B_k x^{t}\right)
\end{split}
\end{align}
where the set $X$ is given by $X:=\{x\mid \|x\|^2\le 1\}$.
It is easy to show that $\|\tilde{\nabla} f(x^{t})\| = 0$ implies that $x^{t}$ is a stationary solution for problem \eqref{eq:example} (see, for example, \cite{jiang14complexity,meisam14nips}). However using the proximal gradient alone is not enough here, as we also need to make sure that the primal feasibility gap $\|x^{t}-x^{t}_k\|$ goes to zero.
Therefore in this work we combine the above two criteria and use the following quantity to measure the progress of all three algorithms
\begin{align}
e\left(x^{t}, \{x^{t}_k\}\right) := \max_{k}\frac{\|x^{t}_k-x^{t}\|}{\|x^{t}\|} + \|\tilde{\nabla} f(x^{t})\|. \label{eq:measure}
\end{align}
All the algorithms we tested will be terminated when $e\left(x^{t}, \{x^{t}_k\}\right)$ reaches below $10^{-3}$.

\subsection{The Results}
We first graphically illustrate the convergence behavior of different algorithms. We set $N=500$, $K=10$, $\lambda=0$, $T_k = 5$, $M_k = 100$, $p_k=0.1$ for all $k$. In Fig. \ref{fig:augmented} -- \ref{fig:Stopping}, the progress of the algorithm is shown by the sequences of the augmented Lagrangian $L(x^{t}, \{x^t_k\}; y^t)$ as well as the optimality measure $e\left(x^{t}, \{x^{t}_k\}\right)$. First we observe that as predicted in Lemma \ref{lemma:L_difference2} and Lemma \ref{lemma:L_bounded2}, the augmented Lagrangian generated by the Async-PADMM algorithm is a decreasing and lower bounded sequence. Second we see that the augmented Lagrangian generated by ADMM (or the PADMM) resembles a stair function. The reason is that between two successive updates, the master node has to wait for a few iterations for the slowest nodes to finish the computation. Nothing is done during such waiting period, leading to constant augmented Lagrangian. Third, we see that the ADMM seems to be able to reduce the augmented Lagrangian quickly, but in terms of the overall optimality measure it takes longer to converge compared with Async-PADMM. 

\begin{figure}[htb] {\includegraphics[width=
1\linewidth]{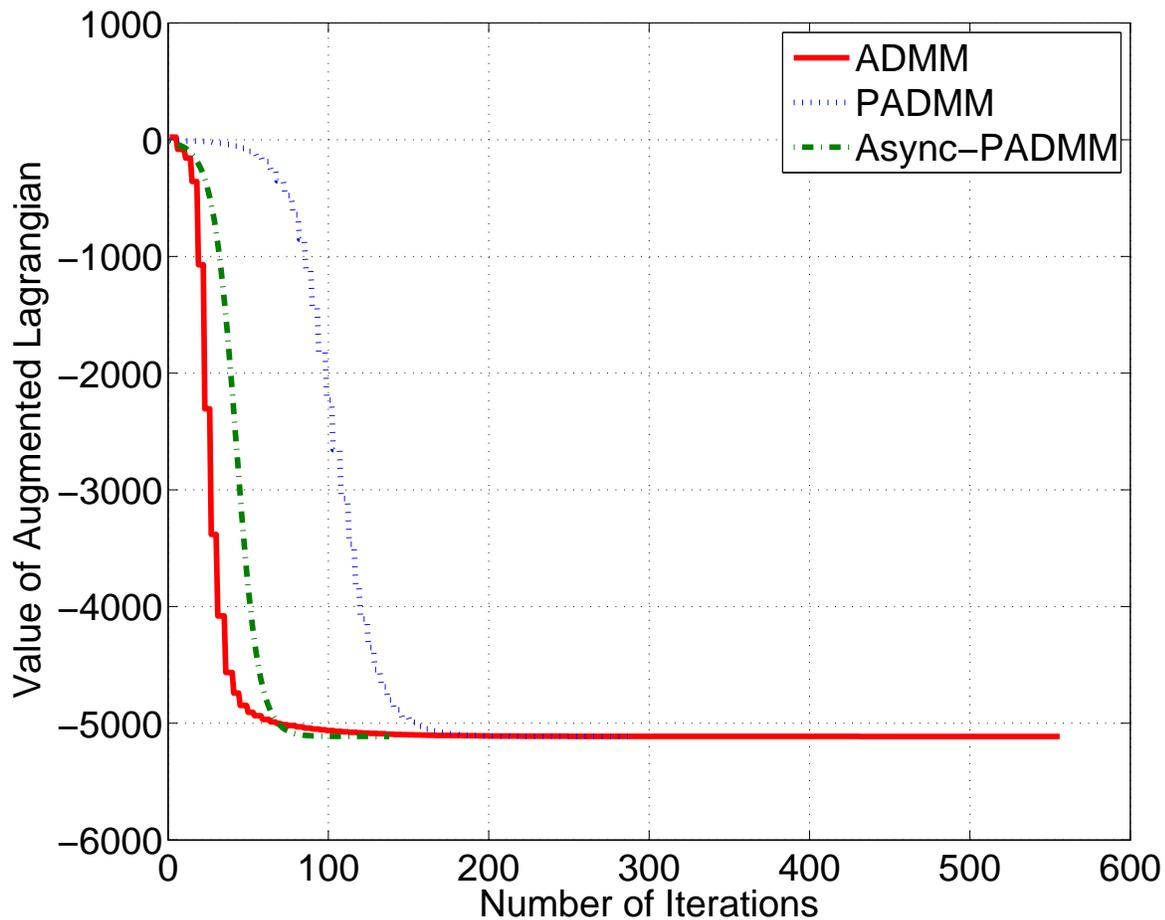}}\vspace{-0.1cm}\caption{Illustration of
the reduction of the augmented Lagrangian for one run of the algorithms. $K=10$, $N=500$, $T_k = 5, \; \forall~k$. For Async-PADMM, $\rho_k$'s are chosen according to \eqref{eq:rho_concave}.}\label{fig:augmented}
\end{figure}

\begin{figure}[htb] {\includegraphics[width=
1\linewidth]{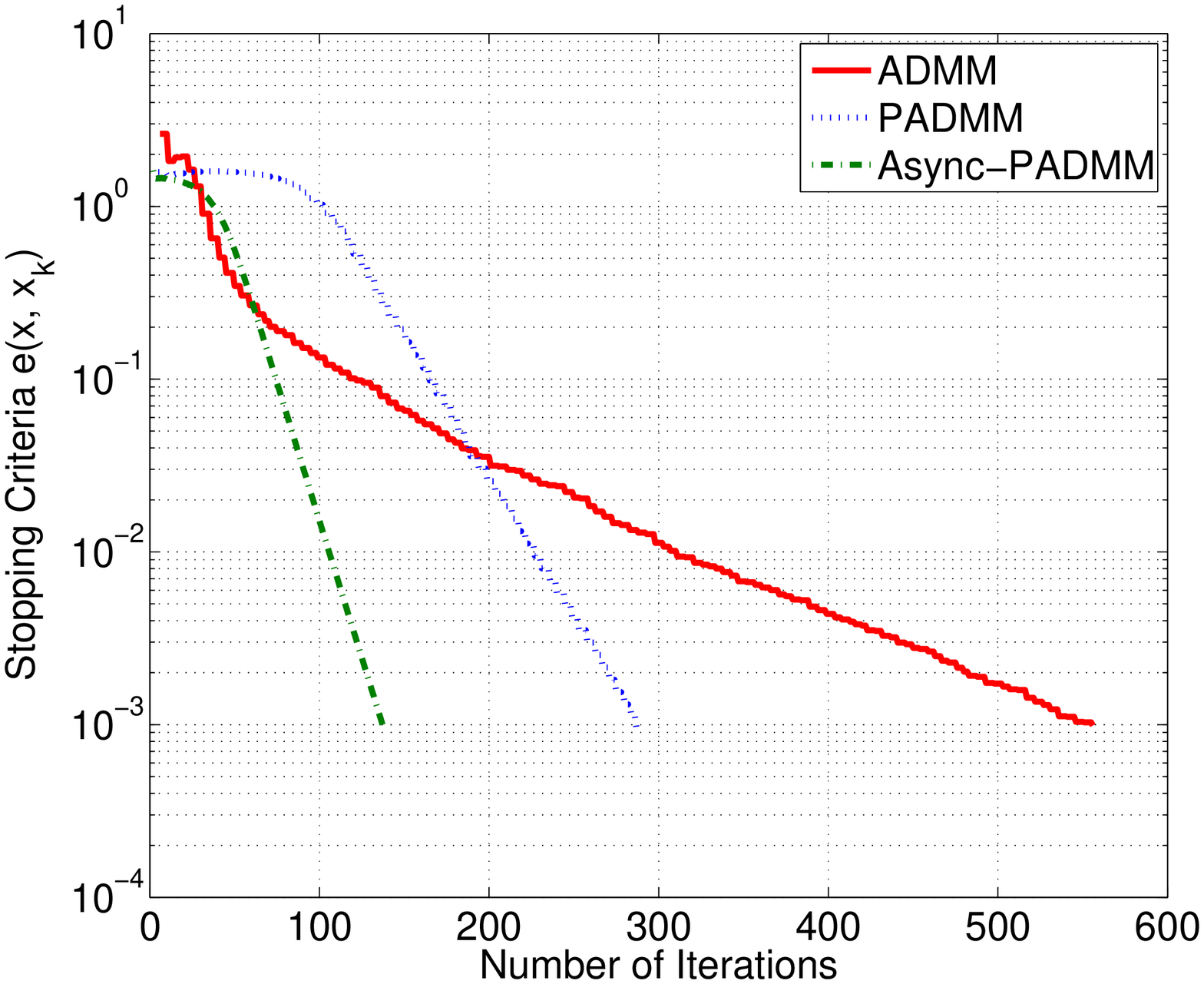}}\vspace{-0.1cm}\caption{Illustration of
the reduction of the stopping criteria for one run of the algorithms. $K=10$, $N=500$, $T_k = 5, \; \forall~k$. For Async-PADMM, $\rho_k$'s are chosen according to \eqref{eq:rho_concave}.}\label{fig:Stopping}
\end{figure}

Next we show the averaged convergence behavior of different algorithms, under various different scenarios. Note that each number in the following table is the average of $50$ independent runs of the respective algorithm.

In Table \ref{table:user}, we compare the behavior of different algorithms with varying number of distributed nodes. We observe that for all three algorithms, the iteration required for convergence is increasing with the number of nodes $K$. It is also clear that the proposed Async-PADMM performs the best (we use underlines to highlight the best result for each scenario).
\begin{table}
\caption{\small {Convergence Behavior When Increasing $K$. $N=500$, $\lambda=0$, $T_k = 5$, $M_k = 100$, $p_k=0.1$}}\vspace*{-0.1cm}
\begin{center}
{\small
\begin{tabular}{|c |c |c |c| c|c| }
\hline
{\bf Method }& {\bf $K=10$}& {\bf $K=20$} & {\bf $K=30$} & {\bf $K=40$} & {\bf $K=50$}\\
\hline
\hline
ADMM & 525& 532& 560& 564& 587\\
\hline
P-ADMM& 362& 456& 536& 636& 625\\
\hline
Async-PADMM& \underline{190}& \underline{220}& \underline{251}& \underline{286}& \underline{285}\\
\hline
\end{tabular} } \label{table:user}
\end{center}
\vspace*{-0.1cm}
\end{table}

In Table \ref{table:delay}, we compare different algorithms under varying degrees of asynchrony. More specifically, in the first four scenarios, we change $T_k$ from $0$ to $9$, for all $k$. In the last two scenarios, we set all the $T_k$'s to be zero, except for Node 10, whose $T_k$ is either $5$ or $10$. This is to test how the algorithms react to a system with a single slow node. We observe that all three algorithms perform well when there is no computational delay (i.e., when $T_k=0$). However once delay starts to increase, ADMM and PADMM become slow, and the transition is quite abrupt (for example ADMM triples its convergence time when $T_k$'s change from $0$ to $3$). This is reasonable as ADMM and PADMM are not designed to deal with asynchrony.
\begin{table*}
\caption{\small {Convergence Behavior For Different Levels of Asynchrony. $N=500$, $K=10$, $\lambda=0$, $M_k = 100$, $p_k=0.1$}}\vspace*{-0.1cm}
\begin{center}
{\small
\begin{tabular}{|c |c |c |c| c|c|c | }
\hline
{\bf Method }& {\bf $T_k=0$}& {\bf $T_k=3$} & {\bf $T_k=6$} & {\bf $T_k=9$} & {\bf $\{T_k=0\}_{k=1}^{9}$} & {\bf $\{T_k=0\}_{k=1}^{9}$} \\
& & & & & {\bf $T_{10}=5$} & {\bf $T_{10}=10$}\\
\hline
\hline
ADMM & 95& 461& 533& 582& 528 & 914 \\
\hline
P-ADMM& 72& 266& 415& 556& 280 & 456 \\
\hline
Async-PADMM& \underline{71}& \underline{113}& \underline{248}& \underline{453}& \underline{98} & \underline{183}\\
\hline
\end{tabular} } \label{table:delay}
\end{center}
\vspace*{-0.1cm}
\end{table*}
In the last set of experiments, we increase the dimensions of the unknown variables and the value of penalization parameter $\lambda$. The results are in Tables \ref{table:dimension} -- \ref{table:lambda}. Again we see that the proposed algorithm works well in both cases.
\begin{table*}
\caption{\small {Convergence Behavior When Increasing $N$. $K=10$, $\lambda=0$, $T_k = 5$, $M_k = 100$, $p_k=0.1$}}\vspace*{-0.1cm}
\begin{center}
{\small
\begin{tabular}{|c |c |c |c| c|c| }
\hline
{\bf Method }& {\bf $N=200$}& {\bf $N=400$} & {\bf $N=600$} & {\bf $N=800$} & {\bf $N=1000 $}\\
\hline
\hline
ADMM & 524& 508& 516& 549& 575\\
\hline
P-ADMM& 360& 361& 357& 392& 379\\
\hline
Async-PADMM& \underline{187}& \underline{180}& \underline{184}& \underline{196}& \underline{188}\\
\hline
\end{tabular} } \label{table:dimension}
\end{center}
\vspace*{-0.1cm}
\end{table*}

\begin{table*}
\caption{\small {Convergence Behavior When Increasing $\lambda$. $K=10$, $N = 500$, $T_k = 5$, $M_k = 100$, $p_k=0.1$}}\vspace*{-0.1cm}
\begin{center}
{\small
\begin{tabular}{|c |c |c |c| c|c| }
\hline
{\bf Method }& {\bf $\lambda=20$}& {\bf $\lambda=40$} & {\bf $\lambda=60$} & {\bf $\lambda=80$} & {\bf $\lambda=100 $}\\
\hline
\hline
ADMM & 531& 570& 544& 616& 916\\
\hline
P-ADMM& 362& 369& 461& 488& 555\\
\hline
Async-PADMM& \underline{194}& \underline{209}& \underline{214}& \underline{274}& \underline{324}\\
\hline
\end{tabular} } \label{table:lambda}
\end{center}
\vspace*{-0.1cm}
\end{table*}

\section{Conclusion}\label{sec:conclusion}

In this paper, we propose an ADMM-based algorithm that is capable of solving the nonsmooth and nonconvex problem \eqref{eq:consensus} in a distributed, asynchronous and incremental manner. We show that as long as the stepsize of the primal and dual updates are chosen sufficiently large, the algorithm converges to the set of stationary solutions of the problem. Numerically we show that the proposed algorithm can efficiently deal with the asynchrony arises from distributedly solving certain sparse PCA problem. In the future, we are interested in analyzing the iteration complexity of the algorithm proposed in this paper. That is, we want to bound the maximum number of iterations needed to reach an $\epsilon$-stationary solution for problem \eqref{eq:consensus}. To the best of our knowledge, such iteration complexity analysis for nonconvex ADMM-type algorithm is not available yet. We are also interested in extending the analysis to asynchronous ADMM algorithm {\it without} using the proximal step. Our current analysis is critically dependent on the availability of the gradient information for each component function, therefore cannot be directly applied to the aforementioned case.

\section{Acknowledgement}
The author wish thank Zhi-Quan Luo from University of Minnesota, and
Tsung-Hui Chang from National Taiwan University of Science and Technology, and Xiangfeng Wang from East China Normal University, for helpful discussions.

\bibliographystyle{IEEEbib}

\bibliography{ref,biblio}

\end{document}